\documentclass{article}
\usepackage{spconf,amsmath,graphicx}

\usepackage{amsthm}
\usepackage{amssymb}
\usepackage{xypic}
\theoremstyle{plain}

\newtheorem{thm}{Theorem}
\newtheorem{prop}[thm]{Proposition}
\newtheorem{cor}[thm]{Corollary}
\newtheorem{lem}[thm]{Lemma}

\theoremstyle{definition}
\newtheorem{df}[thm]{Definition}

\newtheorem{eg}[thm]{Example}

\def\shf{\mathcal}

\def\rank{\textrm{rank }}
\def\image{\textrm{image }}
\def\st{\textrm{star }}
\def\cl{\textrm{cl }}
\def\active{\textrm{active}}
\def\roi{\textrm{roi }}

\title{Analyzing wireless communication network vulnerability with homological invariants}
\name{Michael Robinson}
\address{Mathematics and Statistics\\
American University\\
Washington, DC, USA\\
michaelr@american.edu}

\begin{document}
\maketitle
\begin{abstract}
This article explains how sheaves and homology theory can be applied to simplicial complex models of wireless communication networks to study their vulnerability to jamming.  It develops two classes of invariants (one local and one global) for studying which nodes and links present more of a liability to the network's performance when under attack.
\end{abstract}
\section{Introduction}
Wireless communication networks are vulnerable to adversarial jamming because they are based on a broadcast medium.  However, the problems of designing attack and defense strategies for jamming have been neglected areas of study \cite{Commander_2007}.  Even if the network's nodes and environment remains fixed, which links present the network with the greatest liability?  We will study two classes of networks: (1) a fixed network in which all nodes communicate by broadcasting into a single channel, or (2) a fixed network of base station nodes that communicate with dedicated mobile terminals using a single broadcast resource.  We specifically do not address wired networks, since the vulnerability of wired networks is better understood.

\section{Contributions}
This article argues that appropriate tools for addressing these problems come from algebraic topology, and makes three main contributions to the study of a wireless communication network's physical vulnerability:
\begin{enumerate} 
\item It develops two simplicial complex models of communication networks.  This specifically generalizes graph-based models and permits more refined, higher-dimensional resource contention to be captured. 
\item On these models, the article outlines a sheaf-theoretic model of network activity and interference, which provides a disciplined way to explain the possible non-interfering configurations of active nodes.
\item Finally, the article pioneers vulnerability assessment using persistent and relative homology.  These two techniques yield novel global and per-node (respectively) vulnerability assessments.
\end{enumerate}

\section{Historical context}
We assume that the proper operation of the network under study requires every node to be able to communicate with every other node, possibly through a relay.  Following \cite{Noubir_2004,Gueye_2010}, we will declare that the network is \emph{vulnerable} at a collection of links if their removal results in a disconnected network.  Although this is a fairly drastic mode of failure (the loss of a few links can dramatically reduce the network capacity \cite{Arulselvan_2009}), it is a failure that is easy to describe.

Although vulnerability to purely random failures have been extensively studied, Commander \emph{et al.} \cite{Commander_2007} observe that vulnerability to an adversarial jammer has received little attention.  Xu \emph{et al.} \cite{Xu_2005} was one of the few works discussing jammers, from both the perspective of the attacker and defender.

Although the use of graph theory is well-established in studying resource conflict in wireless networks \cite{Nandagopal_2000,Yang_2002,Jain_2003,Lee_2007}, we are inspired by the detailed survey \cite{Chiang_2007}, which states, ``Solving generalized NUM [network utility maximization] problems over networks with randomly varying topology remains an under-explored area with little known results on models or methodologies.''  Along these lines, graph theory plays a centrol role in identifying \emph{critical nodes} \cite{Arulselvan_2009}, those through which a substantial amount of traffic passes.  Because identifying these nodes is computationally difficult \cite{DiSumma_2011,dinh2012}, our (more na\"ive) definition of vulnerability as causing disconnectedness avoids a combinatorial explosion.

The tools of topology have been used more extensively in sensor networks.  Our simplicial complexes are inspired by the work of \cite{DeSilva_2005,Ghrist_2005,DeSilva_2007,Kanno_2009}.  Jamming and defense problems for sensor networks were discussed in \cite{Ahmed_2005,Xu_2006,Yao_2009}.  Finally, we observe that communication network capacity has been successfully studied in \cite{Ghrist_2011} using the tools of sheaf theory.

\section{Two simplicial models of wireless networks}

We will take the number of path components of a network as a measure of its health: fewer components is better.  (If the network is given a topological structure $X$, as will be done in the next section, it is well known that the zeroth homology group $H_0(X)$ specifies the path components of $X$.)  If there are multiple path components, there exist pairs of nodes which cannot communicate with each other, even through relays.  We also make the following \emph{single channel assumption}: if a link connected to a node is jammed, then that node cannot receive transmissions from \emph{any} other node.  

\begin{df}
A wireless network \emph{vulnerability} is its susceptibility to becoming disconnected when a single source of interference is present.
\end{df}

This article presents techniques for identifying these vulnerabilities within a simplicial model of the network.  We begin with a few preliminary definitions that are relevant to simplicial complexes. 

\begin{df}
An \emph{abstract simplicial complex} $X$ on a set $A$ is a collection of ordered subsets of $A$ that is closed under the operation of taking subsets.  
\begin{itemize} 
\item We call each element of $X$ a \emph{cell}.  A cell with $k+1$ elements is called a $k$-cell, though we usually call a $0$-cell a \emph{vertex} and a $1$-cell an \emph{edge}.  

\item If $a,b$ are cells with $a \subset b$, we say that $a$ is a \emph{face} of $b$, and that $b$ is a \emph{coface} of $a$.  A cell of $X$ that has no cofaces is called a \emph{facet}.  

\item The \emph{closure} $\cl S$ of a set $S$ of cells in $X$ is the smallest abstract simplicial complex that contains $S$.  Observe that $\cl S \subseteq X$.

\item The \emph{star} $\st S$ of a set $S$ of cells in $X$ is the set of all cells that have at least one face in $S$.  Although $\st S \subseteq X$, it is typically not an abstract simplicial complex.
\end{itemize}
\end{df}

Suppose a radio network is active in a spatial region $R$, which we model as being a topological space.  Let the network consist of a collection of nodes $N=\{n_i\} \subset R$ that communicate through a single-channel, broadcast resource.  An open set $U_i$ is associated to each node $n_i$ that represents its \emph{transmitter coverage region}.  (See Figure \ref{fig:ilcomplex}.) For each node $n_i$, a continuous function $s_i : U_i \to \mathbb{R}$ represents its \emph{signal level} at each point in $U_i$.  For simplicity, we assume that there is a global threshold $T$ for accurately decoding the transmission from any node.

Two simplicial complexes are useful for describing the network: the \emph{interference complex} that describes interference received by mobile terminals communicating with each node, and the \emph{link complex} that describes the possible communication links between nodes.

\begin{df}
The \emph{interference complex} $I=I(N,U,s,T)$ consists of all subsets of $N$ of the form $\{i_1, \dotsc, i_n\}$ for which $U_{i_1} \cap \dotsb \cap U_{i_n}$ contains a point $x$ for which $s_{i_k}(x) > T$ for all $k=1, \dotsb n$. 
\end{df}
Briefly, the interference complex describes the lists of transmitters that when transmitting will result in at least one mobile receiver location receiving multiple signals simultaneously.  (Without the constraint on the decoding threshold, the interference complex reduces to the well-known \v{C}ech complex \cite{Hatcher_2002}.)

\begin{prop}
Each facet of the interference complex corresponds to a maximal collection of nodes that mutually interfere.
\end{prop}
\begin{proof}
Let $c$ be a cell of the interference complex.  Then $c$ is a collection of nodes whose footprints have a nontrivial intersection.  By construction, the decoding threshold is exceeded for all nodes at some point $x$ in this intersection.  Therefore, if any two nodes in $c$ transmit simultaneously, they will interfere at $x$.  If $c$ is a facet, it is contained in no larger cell, so it is clearly maximal.
\end{proof}

\begin{df}
The \emph{link graph} is the following collection of subsets of $N$:
\begin{enumerate}
\item $\{n_i\} \in N$ for each node $n_i$, and
\item $\{n_i,n_j\}\in N$ if $s_i(n_j) > T$ and $s_j(n_i) >T$.
\end{enumerate}
The \emph{link complex} $L=L(N,U,s,T)$ is the clique complex of the link graph, which means that it contains all elements of the form $\{i_1,\dotsc,i_n\}$ whenever this set is a clique in the link graph.
\end{df}

\begin{prop}
Each facet in the link complex is a maximal set of nodes that can communicate directly with one another (with only one transmitting at a time).
\end{prop}
\begin{proof}
Let $c$ be a cell of the link complex.  By definition, for each pair of nodes, $i,j\in c$ implies that $s_i(n_j) > T$ and $s_j(n_i) >T$.  Therefore, $i$ and $j$ can communicate with one another.  
\end{proof}

\begin{cor}
Facets of the interfence and link complexes represent common broadcast resources.
\end{cor}

\begin{figure}
\begin{center}
\includegraphics[width=3in]{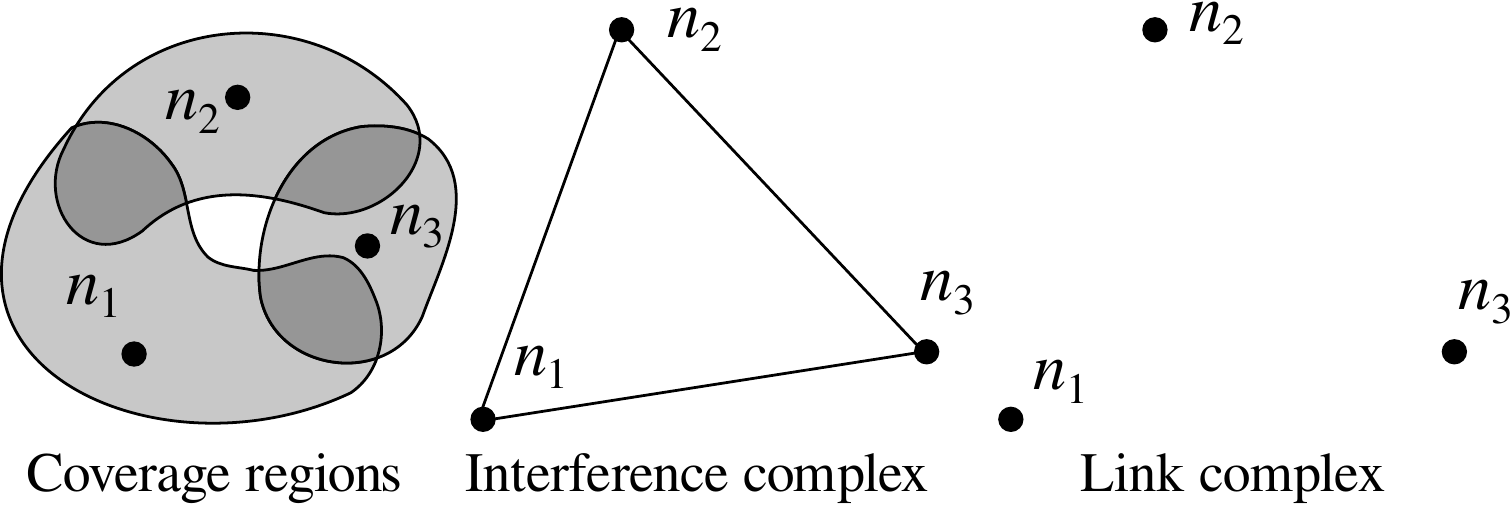}
\caption{Transmitter coverage regions (left), the associated interference complex (middle), and link complex (right)}
\label{fig:ilcomplex}
\end{center}
\end{figure}

\begin{eg}
Figure \ref{fig:ilcomplex} shows a network with three nodes.  The coverage regions are shown at left for a given threshold $T$.  Since there are nonempty pairwise intersections between the coverage regions, but there is no common point of intersection for all three nodes, the interference complex (middle) contains no $2$-cells.  However, since no pair of nodes can actually communicate, the link complex consists of three isolated vertices.
\end{eg}

\section{Interference from a transmission}

The interference caused by a transmission impacts the usability of the network outside of the transmission's immediate vicinity.  This section builds a consistent definition of the \emph{region of influence} of a node or a link within the network.  To justify this definition, we use a local model that describes which configurations of nodes can transmit simultaneously.

\begin{df}
Suppose that $X$ is a simplicial complex (such as an interference or link complex) whose set of vertices is $N$.  Consider the following assignment $\shf{T}$ of additional information to capture which nodes are transmitting and decodable:
\begin{enumerate}
\item To each cell $c\in X$, assign the set
\begin{eqnarray*}
\shf{T}(c)&=&\{n \in N : \text{there exists a cell } d\in X \text{ with }\\
&& c \subset d \text{ and } n\in d\}\cup\{\perp\}
\end{eqnarray*} 
of nodes that have a coface in common with $c$, along with the symbol $\perp$.  We call $\shf{T}(c)$ the \emph{stalk} of $\shf{T}$ at $c$.
\item To each pair $c \subset d$ of cells, assign the \emph{restriction function}
\begin{equation*}
\shf{T}(c\subset d)(n) =
\begin{cases}
n&\text{if }n\in \shf{T}(d)\\
\perp&\text{otherwise}\\
\end{cases}
\end{equation*}
\end{enumerate}
\end{df}

For instance, if $c \in X$ is a cell of a link complex, $\shf{T}(c)$ specifies which nearby node is transmitting and decodable, or $\perp$ if none are.  The restriction functions relate the decodable transmitting nodes at the nodes to which nodes are decodable along an attached wireless link.  Similarly, if $c \in X$ is a cell of an interference complex, $\shf{T}(c)$ also specifies which nearby node is transmitting, and effectively locks out any interfering transmissions from other nodes.  

\begin{df}
The assignment $\shf{T}$ is called the \emph{transmission sheaf} and is an example of a \emph{cellular sheaf} -- a mathematical object that stores local data.  The theory of sheaves explains how to extract consistent information, which in the case of networks consists of nodes that whose transmissions do not interfere with one another.

A \emph{section} of $\shf{T}$ supported on a subset $Y \subseteq X$ is an assignment $s:Y \to N$ so that for each $c \subset d$ in $Y$, $s(c) \in \shf{T}(c)$ and $\shf{T}(c\subset d)\left( s(c)\right) = s(d)$.  A section supported on $X$ is called a \emph{global section}.  
\end{df}

Specifically, global sections are complete lists of nodes that can be transmitting without interference.  

\begin{figure}
\begin{center}
\includegraphics[width=3in]{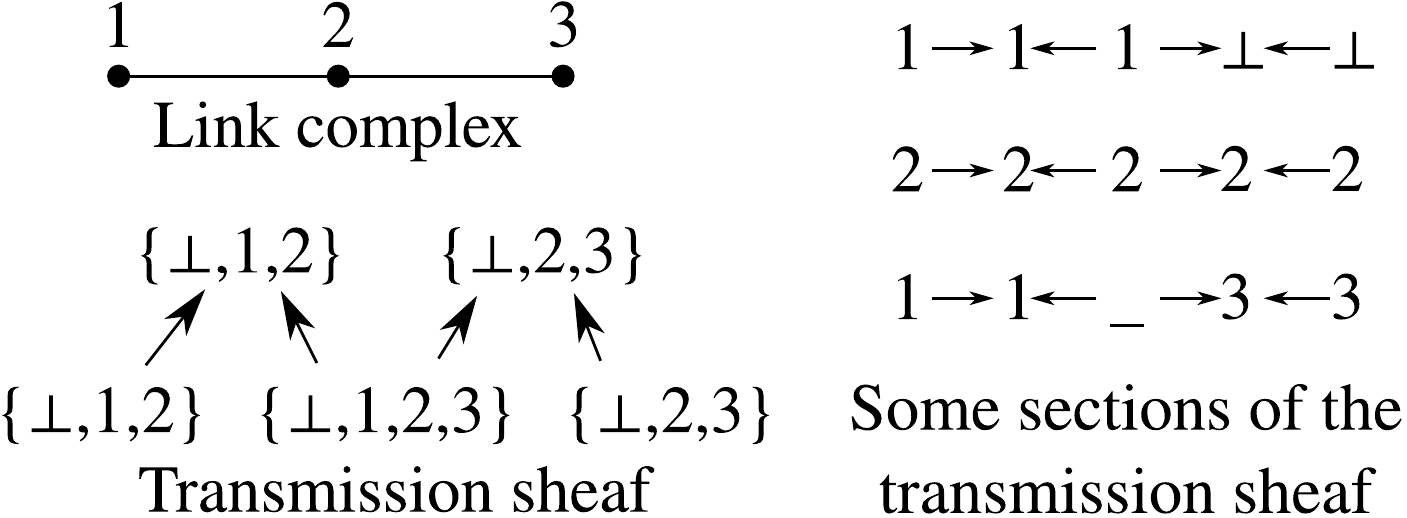}
\caption{A link complex (left top), sheaf $\shf{T}$ (left bottom), and three sections (right).  The restrictions are shown with arrows.  global section when node 1 transmits (right top), global section when node 2 transmits (right middle), and a local section with nodes 1 and 3 attempting to transmit, interfering at node 2 (right bottom)}
\label{fig:linesec}
\end{center}
\end{figure}

\begin{eg}
Figure \ref{fig:linesec} shows a network with three nodes, labeled 1, 2, and 3.  When node 1 transmits, node 2 receives.  Because node 2 is busy, its link to node 3 must remain inactive (right top).  When node 2 transmits, both nodes 1 and 3 receive (right middle).  The right bottom diagram shows a local section that cannot be extended to the cell marked with a blank.  This corresponds to the situation where nodes 1 and 3 attempt to transmit but instead cause interference at node 2. 
\end{eg}

\begin{eg}
Observe that in either of the simplicial complex models shown in Figure \ref{fig:ilcomplex}, only one of the nodes may transmit at a time.
\end{eg}

\begin{df}
Suppose that $s$ is a global section of $\shf{T}$.  The \emph{active region} associated to a node $n\in X$ in $s$ is the set
\begin{equation*}
\active (s,n) = \{a \in X : s(a)=n\},
\end{equation*}
which is the set of all nodes that are currently waiting on $n$ to finish transmitting.
\end{df}

\begin{lem}
\label{lem:act}
The active region of a node is a connected, closed subcomplex of $X$ that contains $n$.
\end{lem}
\begin{proof}
Consider a cell $c \in \active (s,n)$.  If $c$ is not a vertex, then there exists a $b \subset c$; we must show that $b\in \active (s,n)$.  Since $s$ is a global section $\shf{T}(b\subset c)s(b)=s(c)=n$.  Because $s(c) \not= \perp$, the definition of the restriction function $\shf{T}(b\subset c)$ implies that $s(b)=n$.  Thus $b\in \active (s,n)$ so $\active (s,n)$ is closed.

If $c \in \active (s,n)$, then $c$ and $n$ have a coface $d$ in common.  Since $s$ is a global section $s(d)=\shf{T}(c \subset d)s(c)=\shf{T}(c \subset d)n=n$.  Thus, $n \in \active(s,n)$, because $n$ is a face of $d$ and $\active (s,n)$ is closed.  This also shows that every cell in $\active (s,n)$ is connected to $n$.
\end{proof}

\begin{lem}
\label{lem:stactive}
The star over the active region of a node does not intersect the active region of any other node.
\end{lem}
\begin{proof}
Let $c\in \st \active (s,n)$.  Without loss of generality, assume that $c \notin \active (s,n)$.  Therefore, there is a $b \in \active (s,n)$ with $b\subset c$.  By the definition of the restriction function $\shf{T}(b\subset c)$, the assumption that $c\notin \active (s,n)$, and the fact that $s$ is a global section, $s(c)$ must be $\perp$.
\end{proof}

\begin{lem}
The active region of a node is independent of the global section.  More precisely, if $r$ and $s$ are global sections of $\shf{T}$ and the active regions associated to $n \in X$ are nonempty in both, then $\active (s,n)=\active (r,n)$.
\end{lem}
\begin{proof}
Without loss of generality, we need only show that $\active (s,n) \subseteq \active (r,n)$.  If $c \in\active(s,n)$, there must be a cell $d\in X$ that has both $n$ and $c$ as faces.  Now $s(n)=r(n)=n$ by Lemma \ref{lem:act}, which means that $r(d)=\shf{T}(n \subset d)r(n)=n$.  Therefore, since $\active (r,n)$ is closed, this implies that $c \in \active(r,n)$.  
\end{proof}

\begin{df}
Because of the Lemmas, we call the star over an active region associated to a node $n$ the \emph{region of influence}.  The region of influence of a node can be written as a union
\begin{equation*}
\roi n = \bigcup_{f \in F} \st \cl f,
\end{equation*}
for a collection $F$ of facets, so we call the star over a closure of a facet the \emph{region of influence} of that facet.
\end{df}

\begin{cor}
\label{cor:unaffected}
The complement of the region of influence of a facet is a closed subcomplex.
\end{cor}

\section{Global vulnerability assessment using persistent homology}

Based on the previous section, we assume that interfering with a link removes it and its region of influence from the network.  If the strength of the interference is unknown, it is more appropriate to assume that there is a filtration $\emptyset = L_0 \subseteq L_1 \subseteq \dotsc$ of open subcomplexes of the network $X$ that represent an increasingly large collection of links being disrupted by the interference.  At one end of the filtration $L_0 = \emptyset$, in which the network is unaffected.  At the other end, there can be severe disruptions of the network.

Observe that if $R_1 < R_2$, then $L_{R_1} \subseteq L_{R_2}$.  This induces a homomorphism $H_k(X \,\backslash \,\roi L_{R_2}) \to H_k(X \,\backslash \,\roi L_{R_1})$.  

\begin{df}
Because of this, we usually think of elements of $H_k(X \,\backslash\,\roi L_R)$ as being part of a larger \emph{$k$-th persistent homology group} $H_k(X \,\backslash\,\roi L_R)$ associated to the interference filtration $L_R$.  Following standard practice \cite{Edelsbrunner_2002,Ghrist_2008}, an element of $a\in H_k(X \,\backslash\,\roi L_R)$ is said to be \emph{born} at $R_1$ if $a\in H_k(X \,\backslash\,\roi L_{R_1})$ is nontrivial, and its preimage is trivial in all $H_k(X \,\backslash\,\roi L_{R_2})$ with $R_2<R_1$.  Conversely, it is said to \emph{die} at $R_1$ if its image in $H_k(X \,\backslash\,\roi L_{R_2})$ is zero for all $R_2 > R_1$.  The difference between its birth and death is called its \emph{age}.  
\end{df}

Elements of the persistent homology group with larger ages represent ``more significant'' topological features in the data, because they are more robust to perturbations \cite{CohenSteiner_2007}.

\begin{figure}
\begin{center}
\includegraphics[width=3in]{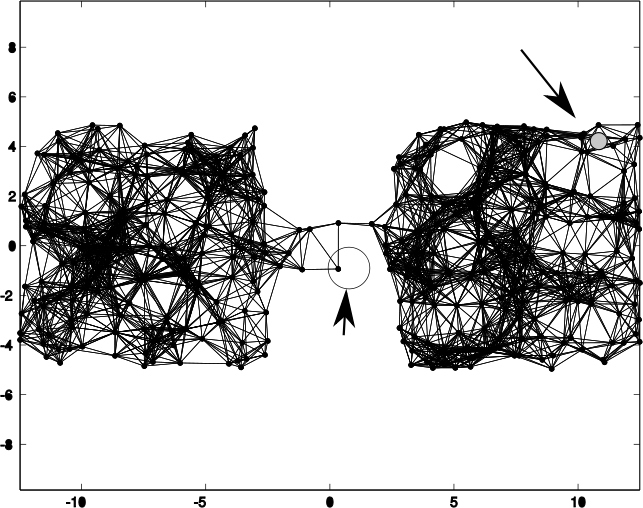}
\caption{A simulated network with two potential inferers marked with arrows: (1) a large empty circle, (2) a solid circle}
\label{fig:netmodel}
\end{center}
\end{figure}

We considered the link complex of a random network $N\subset \mathbb{R}^2$ (Figure \ref{fig:netmodel}) and two interference sources located at $a_1,a_2 \in \mathbb{R}^2$.  To simulate varying levels of interference, the filtration $L_{R,i} = \{n \in N : d(n,a_i) < R \}$ was computed according to the Euclidean distance function $d$.  Given this information, Perseus \cite{Nanda_2013} was used to compute the $0$-th persistent homology of the remaining network for each of the two jammers.  This software decomposes the persistent homology group into its set of generators, and displays the birth and death times corresponding to each.  

Examination of the results (Figure \ref{fig:ph_diagram}) shows that the number of significant generators (farther from the diagonal) of $H_0(X \,\backslash\,\roi L_R)$ depends strongly on the position of the jammer.  In the case where the jammer is located near the narrow portion of the network at the center, there is a significant persistent generator.  Communication from one half of the network to the other is impossible when the jammer is active.  For the other case, when the jammer is located away from the narrow portion of the network, there are no significant generators.  For a large range of interference strengths, the network remains connected.  Because of this, the presence of \emph{more and older persistent generators indicates greater vulnerability to that particular interference source}.

\begin{figure}
\begin{center}
\includegraphics[width=3in]{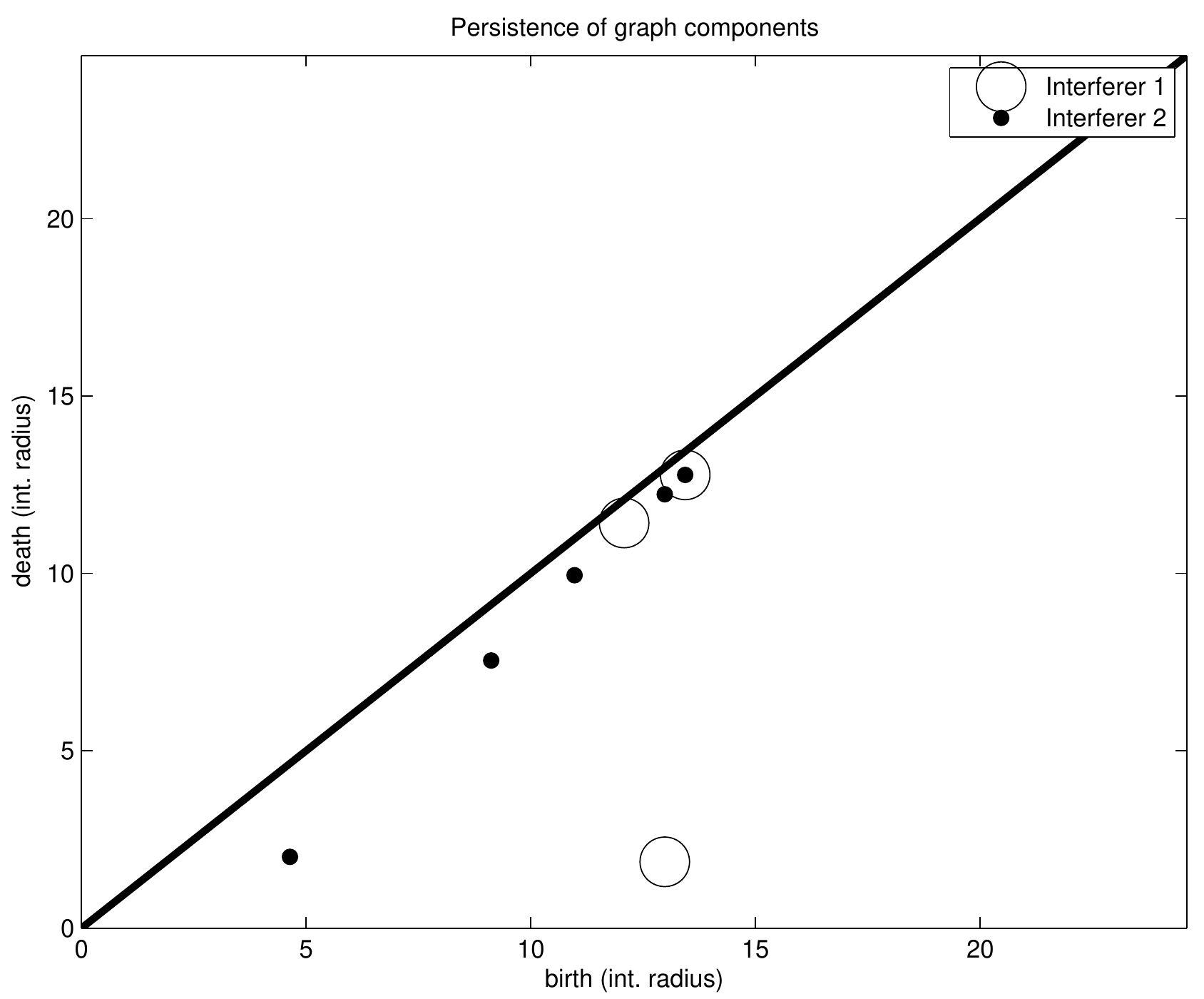}
\caption{Persistence diagram (horizontal=birth radius, vertical=death radius) for connected components associated to the two jammers.  Larger distance from the diagonal signifies the network is more vulnerable to a jammer}
\label{fig:ph_diagram}
\end{center}
\end{figure}

\section{Local vulnerability assesment using relative homology}

The single channel assumption means that an attack on a facet $L$ removes its region of influence from the network.  When this occurs, the nodes that are faces of this facet cannot communicate, but other portions of the network may become disconnected, too.  This section develops a precise, theoretical bound on the number of components a network is split into under the influence of interference.

\begin{figure}
\begin{center}
\includegraphics[width=3in]{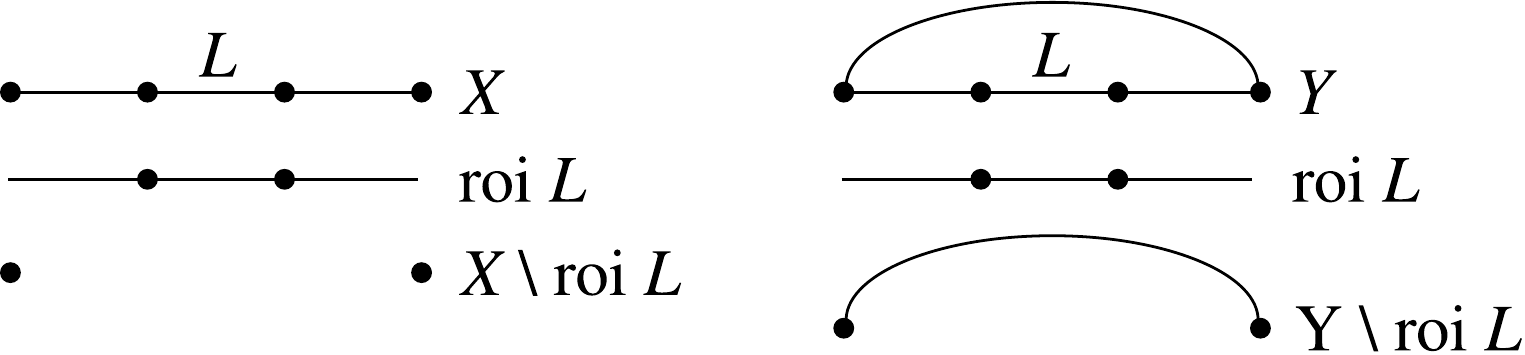}
\caption{Attacking a facet $L$ in a straight-line (left) network $X$ results in a network $X \,\backslash\,\roi L$ that is disconnected, but in a ring network $Y$ (right) does not disconnect the network}
\label{fig:attackline}
\end{center}
\end{figure}

\begin{eg}
A facet $L$ is marked in each of the two networks shown in Figure \ref{fig:attackline}.  If the region of influence of $L$ is removed from the network in the network $X$ at left, it becomes disconnected.  But if removed from the network $Y$ at right, it remains connected. 
\end{eg}

\begin{thm}
\label{thm:bound}
Suppose that $X$ is either a link or interference complex, and that $L$ is a facet of $X$.  If $X$ is connected and $\roi L$ is a proper subset of $X$, the number $\rank H_1(X,X\,\backslash\, \roi L) + 1$ is an upper bound on the number of connected components that an attack on $L$ cuts the network into.  When $H_1(X)$ is trivial, that upper bound is attained.
\end{thm}

\begin{proof}
By Corollary \ref{cor:unaffected}, the portion of the network that is not directly disrupted by the attack (the complement of $L$'s region of influence) is closed.  Observe that this region is the simplicial complex $Y=X\,\backslash\, \roi L$.  Consider the long exact sequence associated to the pair $(X,Y)$, which is
\begin{equation*}
\xymatrix{
\dotsb \ar[r]& H_1(Y) \ar[r] & H_1(X) \ar[r] & H_1(X,Y) \ar[llld] \\
 H_0(Y) \ar[r]& H_0(X) \ar[r]& H_0(X,Y) \ar[r] & 0.
}
\end{equation*}
$H_k(X,Y) \cong \tilde{H}_k(X/Y)$ follows via \cite[Prop. 2.22]{Hatcher_2002}.  This means that $H_0(X,Y)=0$ and $H_0(X)\cong \mathbb{Z}$ because $X$ is connected.  Thus the long exact sequence reduces to
\begin{equation*}
\dotsb \to H_1(Y) \to H_1(X) \to H_1(X,Y) \overset{f}{\to} H_0(Y) \overset{g}{\to} \mathbb{Z} \to 0.
\end{equation*}
The number of connected components of the network during the attack is $\rank H_0(Y)$.  Because of the last term in the long exact sequence, this is at least 1.  If $H_1(X)=0$ then $H_1(Y)=0$ also, so the kernel of the homomorphism $g$ is precisely the image of the monomorphism $f$.  Hence $\rank H_1(X,Y) + 1 = \rank H_0(Y)$ as claimed.

On the other hand, if $H_1(X)$ is not trivial, then $H_1(Y)$ may or may not be trivial, depending on exactly where $L$ happens to fall.  By exactness,
\begin{eqnarray*}
\ker g &=& \rank H_0(Y) - 1 \\
           &=& \image f \\
           &=& \rank H_1(X,Y) - \ker f\\
\end{eqnarray*}
where the rank-nullity theorem for finitely generated abelian groups applies in the last equality. The kernel of $f$ may be as large as $\rank H_1(X)$, but it may be smaller.  Thus, we can claim only that
\begin{equation*}
\rank H_0(Y) \le \rank H_1(X,Y) + 1.
\end{equation*}
\end{proof}

Because of the excision property for relative homology, $H_1(X;X\,\backslash\,\roi L)$ only detects the immediate vicinity of the link, and is a generalization of vertex degree to higher dimensional complexes.  

\begin{eg}
If the simplicial complex for a network is a tree, then attacking a node with degree larger than 2 will result in a network with more than two connected components.
\end{eg}

\begin{figure}
\begin{center}
\includegraphics[width=3in]{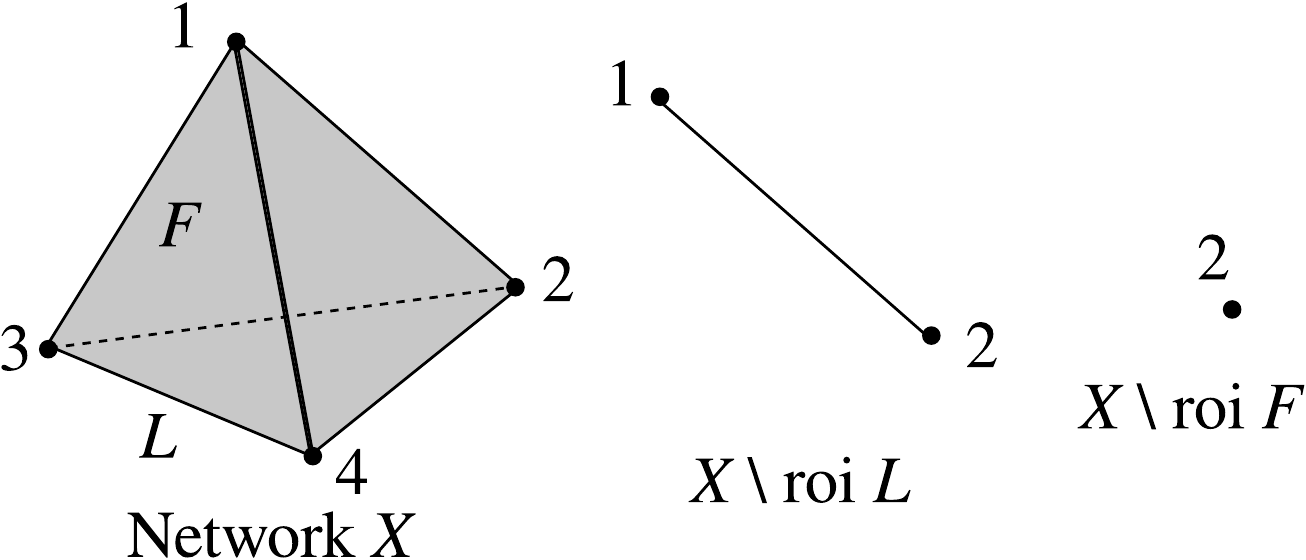}
\caption{An interference complex that is a hollow tetrahedron (left) for which $H_1(X,X\,\backslash\,\roi L)$ and $H_1(X,X\,\backslash\,\roi F)$ are trivial.}
\label{fig:tetra}
\end{center}
\end{figure}

\begin{eg}
Figure \ref{fig:tetra} shows an example of a network for which the interference complex is 2-dimensional.  In this case, $H_1(X,X\,\backslash\,\roi L)$ and $H_1(X,X\,\backslash\,\roi F)$ are trivial.  Although each vertex in the link graph has degree 3, Theorem \ref{thm:bound} implies that the removal of the region of influence for an edge or a face from the network results in a single connected component.
\end{eg}

\begin{eg}
The presence of loops farther from the link can provide back-up connections in the event of an attack, though these are not measured by $H_1(X;X\,\backslash\,\roi L)$.  Therefore, the presence of a nontrivial $H_1(X)$ appears to be protective (but is not a guarantee) of network robustness (see Figure \ref{fig:bubbletree}).
\end{eg}

\begin{figure}
\begin{center}
\includegraphics[width=3in]{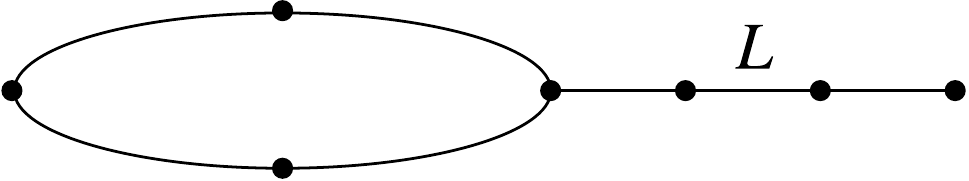}
\caption{Even though the graph for network $X$ has nontrivial $H_1(X)$, attacking facet $L$ results in a disconnected network}
\label{fig:bubbletree}
\end{center}
\end{figure}

\section{Next steps}
This article showed that a wireless communication network's vulnerability to jamming can be analyzed using tools from algebraic topology.  It showed that simplicial complexes, sheaves, and persistent homology provide differing kinds of information about the network.  Our simplicial complex model of a network can include higher dimensional data, but our analysis treated low-dimensional homological data exclusively.  Therefore, it is reasonable to continue our analysis into higher dimensions, aiming to assess decision thresholds and to study the higher homology groups themselves.  

The network vulnerability invariant from relative homology has been studied extensively in the context of \emph{local homology groups}.  However, this mathematical tool has not been applied to networks, so it is worthwhile to examine the impact of the other maps within the long exact sequence, to identify other useful invariants.  Finally, we expect to extend the formalism to timeseries data within a network intending to understanding protocol-level jamming.

\bibliographystyle{IEEEbib}
\bibliography{icassp2014_bib}

\end{document}